\documentclass[copyright,creativecommons]{eptcs}
\providecommand{\event}{TARK 2021} 
\usepackage{breakurl}             
\usepackage{underscore,inputenc}
\usepackage[T1]{fontenc}           

\usepackage{booktabs,pseudocode,amsmath,amsthm,amssymb,sgame}
\usepackage{paralist}
\usepackage{enumitem}
\setlist{topsep=0pt, leftmargin=*}

\usepackage{amsthm,amsmath}

\newtheorem{theorem}{Theorem}
 
\newtheorem{example}{Example}
\newtheorem{lemma}{Lemma}
\newtheorem{definition}{Definition}

\title{Game-Theoretic Models of Moral and Other-Regarding Agents (extended abstract)}
\author{Gabriel Istrate
\institute{West University of Timi\c{s}oara}
\email{gabrielistrate@acm.org}
}
\def\titlerunning{Game-Theoretic Models of Moral Agents}
\def\authorrunning{G. Istrate}
\begin{document}
\maketitle

\begin{abstract}
We investigate Kantian equilibria in finite normal form games, a class of non-Nashian, morally motivated courses of action that was recently proposed in the economics literature. We highlight a number of problems with such equilibria, including computational intractability, a high price of miscoordination, and problematic extension to general normal form games. We give such a generalization based on  concept of \emph{program equilibria}, and point out that that a practically relevant generalization may not exist.  To remedy this we propose some general, intuitive, computationally tractable, other-regarding equilibria that are special cases Kantian equilibria, as well as a class of courses of action that interpolates between purely self-regarding and Kantian behavior. 
\end{abstract}

\section{Introduction}

Game Theory is widely regarded as the main conceptual foundation of strategic behavior. The promise behind its explosive development (at the crossroads of Economics and Computer Science) is that of understanding the dynamics of human agents and societies and, equally importantly, of guiding the engineering of artificial agents, ultimately capable of realistic, human-like, courses of action.  Yet, it is clear that the main models of Game Theory, primarily based on the self-interested, rational actor model, and exemplified by the concept of Nash equilibria, are not realistic representations of the richness of human interactions. Concepts such as \emph{bounded rationality} \cite{simon1997models}, and the limitations they impose on the computational complexity of agents' cognitive models \cite{van2019cognition}  can certainly account for some of this difference. But this is hardly the only possible explanation: People behave differently from ideal economic agents not because they would be irrational \cite{ariely2008predictably}, but since many human interactions are cooperative, rather than competitive \cite{tomasello2009we}, guided by social norms such as \emph{reciprocity, fairness} and \emph{inequity-aversion} \cite{bowles2013cooperative},  often involving \emph{networked minds}, rather than utility maximization performed in isolation  \cite{gintis2016individuality}, driven by moral considerations  \cite{tomasello2016natural} or by other  not purely self-regarding behaviors, e.g. \emph{altruism} \cite{hoefer2013altruism} and \emph{spite} \cite{chen2008altruism,chen2016auction}.

Moral considerations (should) interact substantially with game theory: indeed, the latter field has been used to propose a  reconstruction of moral philosophy \cite{binmore1994game,binmore1994game2,binmore2005natural}; conversely, some philosophers have gone as far as to claim that we need a \emph{moral equilibrium theory} \cite{talbott1998we}. Whether that's true or not,  it is a fact that \emph{homo economicus}, the Nash optimizer of economics, is increasingly complemented by a rich emerging typology of human behavior \cite{gintis2014typology}, that also contains (in Gintis's words) "\textit{homo socialis},  the other-regarding agent who cares about fairness, reciprocity, and the well-being of others, and \textit{homo moralis} \footnote{since our agents are not necessarily human, we will use alternate names such as "moral agent" for this type of behavior.}  ... the Aristotelian bearer of nonconsequentialist character virtues".\footnote{Gintis proposes a taxonomy of behavior with three distinct types of preferences: \emph{self-regarding}, \emph{other regarding} and \emph{universalist}; a further relevant  distinction is between so-called \emph{private} and \emph{public personas}, that leads to further types of behavior such as \textit{homo Parochialis}, \textit{homo Universalis} and \textit{homo Vertus}. See \cite{gintis2014typology} for further details.} 
These claims are well-documented experimentally: for instance, 
Fischbacher et al. \cite{fischbacher2001people} investigated the percent of people having self-regarding preferences in a public goods game, showing that it is in the range of 30-40\%, while the remaining were either other-regarding  or moral agents.  Since artificial agents (will) interact with humans, such concerns are highly relevant to the design of multiagent systems and justify the study of alternative other-regarding notions, e.g. Rong and Halpern's \cite{halpern2010cooperative,rong2013towards} "cooperative equilibria" or \emph{dependency theory} \cite{sichman2002multi,grossi2012dependence}. 
 Other-regarding considerations could be encoded (e.g. \cite{fehr1999theory}) as \emph{externalities} into agents' perceived utilities, that may lead them away from straightforward maximization of their material payoffs.  However, keeping them explicit may be important for agent implementations. 

The purpose of this paper is to contribute to the emerging literature on non-Nashian,  morally inspired game theoretic concepts and, equally important, to bring its concerns and methods to the attention of the various communities represented in TARK. 
We are inspired by what we believe is one of the most intriguing classes of equilibrium concepts that can be seen as morally grounded: \emph{Kantian (a.k.a. Hofstadter) equilibria} \cite{roemer2019we}.  This notion emerged from three separate lines of research converging on an identical mathematical definition, but justifying it, however, from several very different perspectives: \emph{superrationality} \cite{superrational,fourny2020perfect}, \emph{team reasoning} \cite{bacharach1999interactive}, and \emph{Kantian optimization}, respectively \cite{roemer2019we}. 

The common framework (most crisply developed for symmetric coordination games) only considers as relevant the action profiles where all agents choose \emph{the same action}, choosing the action $x$ that, if played by everyone, maximizes agents' (identical) utility functions. The justification of this restriction depends on the perspective: \emph{superrationality} assumes that if rationality constrains an agent to choose a specific course of action $x$, then the same reasoning compels \textbf{all} agents (at least in the case of symmetric games, when all agents are positionally indistinguishable from the original agent) to also choose $x$.\footnote{To cite Hofstadter: "If reasoning dictates an answer, then everyone should independently come to that answer. Seeing this fact is itself the critical step in the reasoning toward the correct answer [...]". Though superrationality does away with the assumption of counterfactual independence of Nash equilibria, it is otherwise compatible with a particular version of homo economicus that requires some very strong assumptions on agent rationality (see \cite{fourny2020perfect} for a discussion).}
In contrast, \emph{Kantian optimization} justifies the limitation to symmetric profiles in a very different manner: Roemer \cite{roemer2010kantian} suggested that agents often ignore the potential for action of the other players, acting instead according to the \emph{Kantian categorical imperative}  \cite{sedgwick2008kant} "act only according to that maxim whereby you can, at the same time, will that it should become a universal law", that is, choose a course of action that, if adopted by every agent, would bring all agents the highest payoff.  \footnote{As recognized by Roemer himself and discussed e.g. in \cite{braham2020kantian}, the connection of Kantian equilibria to actual Kantian ideas is quite loose. Another possible interpretation is that Kantian equilibria embody \textit{rule utilitarianism} \cite{harsanyi1977rule}. Finally, see \cite{sher2020normative} for a discussion of the normative aspects of Kantian equilibria. }  One way to formalize this idea, employed e.g. in \cite{alger2013homo}, is to decouple the \textit{material payoffs} agents receive from their (perceived) utility, which agents maximize in order to select the action. Specifically, assume the given agent $i$ plays strategy $x$ against action profile $y$. We assume that the material payoff the agent receives is $\pi_{i}(x,y)$. On the other hand the utility the agent uses to evaluate alternative $x$ may not be equal to $\pi_{i}(x,y)$ and may in fact, have in fact nothing to do with $y$ at all! Instead, $
u_{i}(x,\textbf{y})=\pi_{i}(x,\overline{x}_{-i}),$ 
where $\overline{x}_{-i}$ is the action profile where all agents other than $i$ play  $x$ as well. That is, the agent evaluates the desirability of action $x$ in isolation from the actions of the other players, as if  choosing $x$ could somehow "magically" determine the other players to adopt the same strategy. \footnote{Frank \cite{frank2010price} refers to this as \emph{voodoo causation}. Elster \cite{elster2017seeing} argues that Kantian optimization seems to be rooted in a form of \emph{magical thinking}, "causing agents to act on the belief (or act as if they believed) that they can have a causal influence on outcomes that are effectively outside their control". We take a descriptive, rather than normative position: such reasoning is something people  \textbf{simply do}; understanding its implications is strategically valuable.} Alternatively,  $\pi(x,\overline{x}_{-i})$ measures the extent to which action $x$ is "the morally best course of action". \textbf{Such a justification is cognitively plausible:} experiments have shown  \cite{levine2020logic} that people often employ such "universalization" arguments when judging the morality of a given behavior. 

The questions we attempt to start answering in this paper are: 
\begin{itemize}
\item[1.] \textit{Can we extend the definition of Kantian equilibria to cover all natural cases of "Kantian behavior"?}

However, we are \textbf{not} simply looking in our generalization for yet another equilibrium notion of primarily mathematical interest, but \underline{for one satisfying specific tractability requirements that ensure} \\ \underline{easy implementation in computational agents.} Specifically, a  target concept should at least be: 
\begin{itemize} 
\item[ I.] \textbf{expressive}, i.e. indicative of realistic behavior of human agents in sufficiently typical situations 
\item[ II.] \textbf{cognitively plausible}: the equilibrium should \textbf{not} be justifiable in terms of expensive epistemic assumptions (the way common knowledge of rationality can be used to justify Nash equilibria \cite{aumann1995epistemic}); 
\item[ III.] \textbf{logically tractable}: proposed equilibria should be easy to specify formally, in a way that translates to efficient implementations. 
\item[ IV.] \textbf{computationally tractable}: equilibria should be easy to compute \cite{van2019cognition}, since bounded rational agents are assumed to compute (and play) them. 
\end{itemize} 
\textbf{The main message of the paper is that such an extension is possible, but any general notion of Kantian equilibria may be of theoretical interest only:} while we give an interesting extension for certain symmetric games (Sec.~\ref{sec:gen}) inspired by the concept of \emph{program equilibria}, it's not clear how to further extend it. Together with intractability (Thm.~\ref{nphard}) this suggests that \textbf{a general, practically relevant, notion of Kantian equilibrium might not exist.} 

\item[2.] \textit{What is the relation between Kantian equilibria and Bacharach's (informally defined) team-reasoning equilibria \cite{bacharach1999interactive}?} The answer is that Kantian equilibria are a proper subset of team-reasoning equilibria. 

\item[3.] \textit{Given that the answer to Q1 is negative, are there more specialized equilibria related to Kantian optimization that satisfy (I)-(IV)?} We will show that there exist, indeed, several more restrictive equilibrium notions, satisfying tractability and plausibility constraints, and relate them to Kantian equilibria. 
 
\item[4.] \textit{Real people are seldom purely selfish or purely Kantian. (How) can we formalize this?} We give such a definition, and motivate it through the case of Prisoners' Dilemma. 
\end{itemize}

The outline of the paper is as follows: In Section~\ref{sec:prelim} we review some basic notions. In Section~\ref{seq:kantian} we obtain some further results on (and highlight some limitations of) Kantian equilibria: first of all, we point out that finding a mixed Kantian equilibrium is computationally intractable even for two-player symmetric  games (Theorem~\ref{nphard}). Second, multiple Kantian equilibria may exist, and lack of coordination on the same equilibrium may be detrimental to players, even with all of them playing a common linear combination of Kantian actions. In Section~\ref{sec:gen} we discuss the problem of extending Kantian equilibria to non-symmetric games. giving a proposal based on the concept of program equilibria. As such, our proposal inherits the problems of this concept. Given these problems, in Section~\ref{sec:efficient} we propose several other-regarding equilibria.\footnote{Generally, ethical egoism and its variant, rational egoism, are not accepted as a basis of moral behavior; counterexamples exist,  \cite{rand1964virtue}; however, it's fair to say that such positions are controversial, and somewhat marginal. In contrast, moral and other-regarding behaviors are better aligned, with other-regarding behavior often a consequence of moral play.} We show (Theorem~\ref{thm:efficient}) that these equilibria can be computed efficiently, that they are indeed Kantian equilbria (according to our generalized definition), and that they yield Kantian equilibria for symmetric coordination games. Finally, in Section~\ref{sec:bounded} we relax the assumption that the agents are other-regarding: we assume that agents have a degree of greed, zero for Kantian agents, infinite for Nashian agents. We show (Theorem~\ref{pdgreed}) how our definition applies to Prisoners' Dilemma. 

For reasons of space, \textbf{most proof details are deferred to a longer version of the paper, available on arXiv \cite{istrate-kantian-tr}}. So have we done, for reasons of abundance of technical details, with some of the results: e.g. the ones the proper definition and characterization of Kantian program equilibria (Theorems~9,~10 in \cite{istrate-kantian-tr}), which also clarify the connection between Kantian and  \emph{team reasoning}.  

\section{Related Work} 

The literature on other-regarding game-theoretic models is quite large, and a short section like this one cannot do justice to all the related, relevant work. Instead we have chosen to highlight a modicum of references directly relevant to our work. 

The major impetus for this work was \emph{Kantian optimization}. It was developed in \cite{roemer2010kantian,roemer2015kantian}, developing early ideas of Laffont \cite{laffont1975macroeconomic}.  The current status of the theory is consolidated in the recent book \cite{roemer2019we}. A recent special issue of the \emph{Erasmus Journal of Economics} is devoted to discussing and situating Roemer's contribution. Particularly valuable articles in this collection include \cite{braham2020kantian,sher2020normative}. 

The other strand of ideas relevant to our work concerns the concept of \emph{program equilibria}, defined in \cite{tennenholtz2004program} and further investigated in \cite{fortnow2009program,kalai2010commitment,van2013program,lavictoire2014program,barasz2014robust,critch2019parametric,oesterheld2019robust}. There are several other related (and relevant) models, such as the \emph{translucent player} model of \cite{capraro2015translucent,halpern2018game}, or \emph{mediated equilibria} \cite{monderer2009strong}.

The two other paradigms leading to the same concept for two-player symmetric coordination games, \emph{superrationality} and (especially) \emph{team reasoning} are, of course, relevant to our approach. Superrationality is rather different, though, and we only reiterate recommendations of \cite{superrational,fourny2020perfect}. The main reference for team reasoning is still \cite{bacharach2006beyond}. We also recommend papers \cite{sugden2003logic,colman2018team,gold2020team}. 

Notions of symmetry in games have been insufficiently investigated, and they play an important role in defining Kantian programs. We refer to \cite{ham2013notions,tohme2019structural} for such studies. 

Finally, an impressive amount of work on behaviorally relevant game-theoretic notions related to moral behavior is summarized in \cite{dhami2016foundations}. While it is by no means comprehensive (especially with respect to the computer science literature), it is an excellent starting point.

\section{Preliminaries} 
\label{sec:prelim}

We assume knowledge of basic results of game theory at the level of a textbook such as, e.g. \cite{osborne:rubinstein:book}, in particular with concepts such as normal form games, best response strategy, and mixed (Nash) equilibria. All the games $G$ we consider are normal form and, unless mentioned otherwise, have identical action sets $Act_{G}$ for all players. Given a finite set $S$, we will define $\Delta(S)$ to be the set of probability distributions on $s$. Elements of $\Delta(S)$ are functions $c:S\rightarrow [0,1]$ satisfying $\sum\limits_{i\in S} c(i)=1$. $\Delta^{n}:=\Delta(\{1,2,\ldots, n\})$ is, geometrically, a $(n-1)$-dimensional simplex. When $G$ is a normal-form game and $k$ a player in the game we will denote by $\Delta_{G}^{k}$ the set of mixed actions available to player $k$, identified with some simplex $\Delta^{n}$ with a suitable dimension. We will occasionally drop $k$ from the notation and simply write $\Delta_G$ instead when the player is clear from the context, or when all agents have the same action set.  Given vectors $x=(x_{1},\ldots, x_{n})$ and $y=(y_{1},\ldots, y_{n})$, we say that \emph{$x$ dominates $y$} iff $x_{i}\geq y_{i}$ for all $i=1,\ldots, n$. The domination is \emph{strict} if at least one inequality is. When comparing (mixed) action profiles \textbf{a} and \textbf{b}, the domination relation may apply to the vectors of agent utilities $(u_{1}(\textbf{a}), \ldots,  u_{n}(\textbf{a}))$ and $(u_{1}(\textbf{b}), \ldots,  u_{n}(\textbf{b}))$, respectively. Action profiles that are strictly dominated may be assumed not to occur in game play. 

A game with identical action sets is \emph{diagonal} if every pure action profile is Pareto dominated by some profile on the diagonal, the set of action profiles where all players play the same action. A particular class of diagonal games are \emph{coordination games}, where all player utilities are zero outside the diagonal. Such a game is \emph{symmetric} if, additionally, agent utilities are identical for all action profiles on the diagonal. 


\begin{definition} 
 Let $G$ be a game with common action set $A$. A \emph{variation function} is a function $\phi:\Xi \times A\rightarrow A$, for some set of parameters $\Xi$.\footnote{The precise form of this  definition follows \cite{sher2020normative}, and is motivated by Roemer's definition of \emph{additive/multiplicative Kantian equilibria}, with action set  $A=\mathbb{R}_{+}$ and variation functions $\phi(r,a)=a+r$, $a\cdot r$, respectively. \textbf{We will mostly be concerned with variation functions of the type "change (everyone's) current action to  b"} (for $b\in A$). Formally, $\Xi=A$ and $\phi(b,a)=b$.} 
 A \emph{Kantian (Hofstadter) equilibrium} is a pure strategy profile $x^{opt}=(x^{opt}_{1},\ldots, x^{opt}_{n})$ that maximizes the material payoff 
of each agent, should everyone deviate similarly. Formally, for every agent $i$ and $r\in \Xi$, 
\[
V_i(x^{opt}_{1}, \ldots, x^{opt}_{n})\geq V_{i}(\phi(r,x^{opt}_{1}), \ldots, \phi(r,x^{opt}_{n})).
\]

\label{def:basic}
\end{definition} 

\begin{example} 
One of the original applications of Kantian equilibria was Prisoners' Dilemma (PD, Fig.~\ref{fig:pd} (a)). Kantian equilibria provide an elegant solution to the paradox: Kantian agents  coordinate on action profile (C,C), as jointly doing so gives them a higher payoff than the Nash equilibrium (D,D). 

\begin{figure} 
\begin{center} 
\begin{minipage}{.40\textwidth}
\begin{game}{2}{2}[Player 1][Player 2]
	    &  C      &  D     \\
	 C  &  $2, 2$ & $0, 3$  \\
	 D  &  $3, 0$ & $1, 1$\\
\end{game}
\end{minipage} 
\begin{minipage}{.40\textwidth}
\begin{game}{2}{2}[Player 1][Player 2]
	    &  B     &  S     \\
	 B  &  $2, 3$ & $0, 0$  \\
	 S  &  $1, 1$ & $3, 2$\\
\end{game}
\end{minipage} 
\end{center} 
\caption{a. Prisoners' Dilemma. b. BoS game as modified by Roemer.} 
\label{fig:pd} 
\end{figure} 

\end{example} 

Kantian equilibria are easiest to justify for symmetric diagonal games, since in this case they  dominate all other action profiles, thus can be properly seen as "best course of action for all".  There are symmetric (nondiagonal) games, though, where no pure strategy Kantian equilibrium is adequate, and which seem to compel us to considering mixed-strategy Kantian equilibria. An example is Hofstadter's "Platonia's Dilemma" \cite{superrational}, a special case of the \emph{market entry games} of Selten and G\"{u}th \cite{selten1982equilibrium}: 

\begin{definition} 
In Platonia Dilemma $n$ agents (say, $n=20$) are offered to win a prize. Agents may choose to send their name to a referee. An agent wins the prize if and only if it is \textbf{the only one
submitting their name}: if zero or at least two agents send their names then noone wins anything. 
\end{definition}

It is easy to see that both pure strategies, sending/not sending their name, are equally bad if adopted by all agents: they get zero payoff. A better option is to allow independent randomization: 

\begin{definition} 
Given a game $G$ with identical actions, a  \textbf{mixed Kantian agent} will choose a mixed strategy $X^{OPT}\in \Delta_G$ that maximizes its expected utility, should everyone play $X$. For two-player symmetric games with game matrix $A$ and variation function $\phi(b,a)=b$ we have  
$X^{OPT}= argmax\{y^{T}Ay: y\in \Delta_{G}\}$. 
\end{definition}

\begin{lemma} In Platonia Dilemma the probabilistic strategy where each agent independently submits their name with probability $p$
brings an expected profit to every agent equal to $p(1-p)^{n-1}$. This quantity is maximized for $p=\frac{1}{n}$. Thus the strategy with $p=\frac{1}{n}$ is a mixed Kantian equilibrium. 
\label{one}
\end{lemma} 

\begin{proof} 
Let $f(p)=p(1-p)^{n-1}$. $f^{\prime}(p)= (1-p)^{n-1}- (n-1)p(1-p)^{n-2}= (1-p)^{n-2} (1-np)$, so $f$ is increasing on $[0,1/n]$ and decreasing on $[1/n,1]$. 
\end{proof} 

\section{Limitations of Mixed Kantian Equilibria}
\label{seq:kantian}

In this section we  note some properties of mixed Kantian equilibria. They are mostly negative: finding a mixed Kantian equilibrium is intractable. Also, such equilibria may be vulnerable to miscoordination. 

\subsection{The Computational Intractability of Mixed Kantian Equilibria}

We make an easy observation concerning the computational complexity of mixed Kantian equilibria in symmetric two player games.  To our knowledge this has not been discussed before. Note: such equilibria are guaranteed to exist, since the $(n-1)$-dimensional simplex of mixed strategies  is a compact set and the common utility function is continuous. First of all, finding a mixed Kantian equilibrium is easy in symmetric coordination games, as all such equilibria coincide with pure Kantian equilibria. 

\begin{theorem} Consider a finite symmetric coordination game. Then mixed Kantian equilibria coincide with pure Kantian equilibria. Hence one can compute mixed Kantian equilibria in polynomial time.
\label{th1}   
\end{theorem} 

Platonia Dilemma with $n=2$ shows that Theorem~\ref{th1} does not extend to general symmetric games. This is no coincidence:  in this case finding (or just detecting) the optimal mixed strategy is intractable: 

\begin{theorem} The following problem, called MIXED KANTIAN EQUILIBRIUM, is NP-hard: 
\begin{itemize}
\item[] INPUT: A two-player symmetric  game $G$, and an aspiration level $r\in \mathbb{Q}$. 
\item[] TO DECIDE: Is there a mixed strategy profile $x=(x_1,\ldots, x_N)$ such that the utility of every player under common mixed action $x_{1}a_1+x_{2}a_2+\ldots + x_{m}a_{m}$ is at least $r$? 
\end{itemize}
\label{nphard} 
\end{theorem} 

\begin{proof} 

We point out to the existence of a reduction from CLIQUE to MIXED KANTIAN EQUILIBRIUM, that shows that the latter problem is NP-hard. In fact the reduction will only consider symmetric games with 0/1 payoffs. 

Consider, indeed, a graph $g$. Let $k$ be an integer and $(g,k)$ be the corresponding instance of CLIQUE. 

Define the symmetric two-player game $G$ whose payoff matrix is the adjacency matrix $A$ of $g$. 

Mixed Kantian equilibria $x=(x_1,\ldots, x_N)$ of $G$ correspond to optimal solutions of the following quadratic program: 
\begin{equation}
\label{program-mixed-kant}  
\left\{\begin{array}{c}
max(x^{T}Ax)\\
x_1+\ldots + x_N =1\\
x_{1},\ldots, x_N \geq 0. 
\end{array} 
\right. 
\end{equation} 

This is a problem that has been called \cite{bomze1998standard} \textit{the standard quadratic optimization problem}, and has been investigated substantially in the global optimization literature (see e.g. \cite{bomze1997evolution}).  A beautiful result due to Motzkin and Straus \cite{motzkin1965maxima} can be restated as claiming that for programs whose matrix $A$ is the adjacency matrix of a graph $g$, if $o$ is the optimum of problem~(\ref{program-mixed-kant}) then $\frac{1}{1-o}$ is the size of the maximum clique in $g$.  

Hence $(g,k)\in CLIQUE$ if and only if $(G,\frac{k-1}{k})\in$ MIXED-KANTIAN-EQUILIBRIUM.

\end{proof}

\subsection{Multiple Equilibria and miscoordination}

Optimal diagonal action profiles may fail to be unique. If the agents are not communicating (and no implicit coordination mechanisms are acting, e.g., one of the action profiles being a \emph{focal point}, such as in the Hi-Lo game from \cite{bacharach2006beyond}), agents may reach a suboptimal action profile due to their lack of coordination on the same optimal action: Consider, indeed,  the game in Figure~\ref{multiple}. $(C,C)$ and $(E,E)$ are equally good pure (and mixed) Kantian equilibria. But if one player plays $C$ and the other plays $E$ the resulting outcomes are the worst possible for both of them, being dominated by every single possible strategy profile! Randomizing among Kantian actions might not help either: miscoordination impacts even "Kantian" scenarios, where players, lacking a salient equilibrium to coordinate on, play a joint mixed strategy formed of Kantian actions.\footnote{Such a scenario is, of course, not justifiable from a usual rational choice perspective. But \textbf{it is justifiable in a Kantian setting where every player believes that choosing a pure action $a$ will immediately make all other players do the same}: a player may use the Kantian imperative to restrict itself to pure Kantian equilibria, then use the assumption  to justify playing a convex combination of pure Kantian equilibria it is indifferent between.} We quantify the degradation in performance as follows: 

\begin{figure} 

\begin{center} 
\begin{minipage}{.25\textwidth}
\begin{game}{3}{3}[Pl1][Pl 2]
	    &  C      &  D   & E   \\
	 C  &  $5, 5$ & $3, 6$  & $1,2$\\
	 D  &  $6, 3$ & $4, 4$ & $6,3$\\
	 E  &  $2, 1$ & $3, 6$ & $5,5$\\
\end{game}
\end{minipage} 
\begin{minipage}{.20\textwidth}  
\begin{game} {2}{2}[][Pl2]
	    &  C      &  D     \\
	 C  &  $10, 1$ & $0, 0$  \\
	 D  &  $0, 0$ & $4, 2$\\
\end{game}
\end{minipage} 
\begin{minipage}{.20\textwidth}
\begin{game}{2}{2}[][Pl2]
	    &  B     &  S     \\
	 B  &  $6, 1$ & $0, 0$  \\
	 S  &  $0, 0$ & $3, 2$\\
\end{game}
\end{minipage} 
\begin{minipage}{.25\textwidth} 
\begin{game}{2}{2}[][Pl2]
	    &  C    &  S    \\
	 C  &  $10, 10$ & $100, 200$  \\
	 S  &  $200, 100$ & $6, 6$\\
\end{game}
\end{minipage}

\end{center} 
\caption{(a). A game with multiple Kantian equilibria. (b). Modified BoS (Example~\ref{ex}). (c). Modified BoS (Example~\ref{bos2}). (d). An anti-coordination game. } 
\label{multiple}
\end{figure}


\begin{definition} For a symmetric game $G$ with strictly positive payoffs let $NC$ be the set of mixed action profiles composed of Kantian actions only. The \emph{price of miscoordination} of  $G$ is the ratio $
p(G)=\sup\limits_{a\in NC} \frac{u_{i}(X^{OPT})}{u_{i}(a)}$.  Because of symmetry this does not depend on the particular choice of player $i$. 
\end{definition}

The following result shows that the price of miscoordination can be arbitrarily large: 

\begin{theorem} Let $G$ be a symmetric diagonal game with $k\geq 2$ players and $r\geq 1$ pure Kantian actions. Then the price of miscoordination of  $G$ is in the range $[1, r^{k-1}]$. Both bounds are tight and can be reached in settings where players choose a Kantian action uniformly at random. 
\label{miscoord} 
\end{theorem}
The merit of this simple result is to point out that the definition of generalized Kantian equilibria needs to include scenarios where randomness is correlated, as in  \emph{correlated equilibria}  (see e.g. \cite{shoham2009multiagent}). 
\begin{proof} 
The price of miscoordination is insensitive to dividing all utilities by the same factor $\lambda$, so w.l.o.g. one may assume that the utilities agent receive on pure Kantian equilibrium profiles is 1. For the mixed action \textbf{a} where players play the $r$ Kantian actions (w.l.o.g. $1,2,\ldots, r$) with probabilities $p_{1},p_{2},\ldots p_{r}$ (which add up to 1), its expected utility is $E[u_{i}(\textbf{a})]=\sum u_i(i_{1},i_{2},\ldots i_{k})\cdot p_{i_1}p_{i_2}\ldots p_{i_k}\geq \sum\limits_{i=1}^{r} p_{i}^{k}\geq r\cdot \frac{1}{r}^{k}=\frac{1}{r^{k-1}},$  by Jensen's inequality. The upper bound is obtained when off-diagonal action profiles formed of Kantian actions only have utilities equal to 0.  As for the lower bound, for diagonal games by domination we have $u_{i}(i_{1},i_{2},\ldots i_{k})\leq 1$, so $E[u_{i}(\textbf{a})]=\sum u_{i}(i_{1},i_{2},\ldots i_{k})\cdot p_{i_1}p_{i_2}\ldots p_{i_k}\leq \sum p_{i_1}p_{i_2}\ldots p_{i_k}= (p_{1}+p_{2}+\ldots p_{r})^{k}=1.$ A game realizing the lower bound is the one where agent utilities on all pure action profiles are equal to 1. 
\end{proof} 

\section{Kantian Program Equilibria in (Pareto) Symmetric Games}
\label{sec:gen} 
Definition~\ref{def:basic} of Kantian equilibria makes the most sense in symmetric coordination games, but does not capture all the intuitive cases of Kantian behavior.  
Indeed, let us consider the BoS game, as modified by Roemer (Fig.~\ref{fig:pd} (b)).\footnote{Roemer (\cite{roemer2019we}, Proposition 2.3) argues that (S,B) is a simple Kantian equilibrium. His argument is, however, ad-hoc, based on making this profile "diagonal" by flipping the order of B and S for the second player, and the conclusion that $(B,S)$ is Kantian is, we feel, unintuitive, since $(B,B),(S,S)$ strictly dominate it. Our protocol plays $\frac{1}{2}(B,B)+\frac{1}{2}(S,S)$, different from (and better than) what Roemer calls the mixed Kantian equilibrium, where row player plays $\frac{3}{8}B+\frac{5}{8}S$ and the column player plays $\frac{3}{8}S+\frac{5}{8}B$.} Intuitively, agents would perhaps agree that the following \textbf{protocol} could be called Kantian, in that it is symmetric and both players benefit if they both follow it: flip a fair coin; if it comes out \emph{heads}, they (both) play $B$, else both play $S$. 
As described, the  protocol requires the centralized choice of a random bit, but it could easily be implemented in a distributed manner by making each of the two agents flip a (fair) coin and taking their XOR. The implementation of the protocol (Algorithm~\ref{alg2}) assumes that each agent is parameterized by an \emph{agent ID} $i\in \{1,2\}$ (not needed in this particular case) and vector $(myb,otherb)$ of random bit choices, one for each player, and shared between players. 

An even more dramatic case is that of the game from Figure~\ref{multiple} (d),  where the best outcomes are not symmetric. In these cases it even seems  irrational for the agents to play symmetric action profiles, since these action profiles are dominated by all the other action profiles!  Rather, it is plausible that agents would agree that they need to \textbf{anti}coordinate, but they have different preferences for the joint action profile to coordinate upon. A "best for all" solution would jointly play a random anti-coordinated profile $\frac{1}{2}(C,S)+\frac{1}{2}(S,C)$. As in the previous example, this course of action can be implemented by the two agents in a distributed manner, by jointly playing according to the protocol in Algorithm~\ref{alg2.2}.  In this example, in addition to the extra bit $otherb$ communicated by the other player, the protocol of each agent \textbf{makes explicit use of the agents' own $id$, $i\in \{1, 2\}$.}  

\begin{minipage}{.45\textwidth}
\begin{pseudocode}\\
{BoS}
BOS(i::ID,myb::BIT,otherb::BIT)\\
\\
\mbox{Randomly choose a bit }myb\in \{0,1\}\\
\mbox{communicate }mybit\mbox{ to the } \\
\mbox{other player as its }otherb.\\
\mbox{\textbf{if }}[myb\oplus otherb == 0]\\
 \hspace{5mm}\mbox{ \textbf{then} play }B\\
\hspace{5mm}\mbox{ \textbf{else} play }S\\
\label{alg2}
\end{pseudocode}
\end{minipage} 
\begin{minipage}{.50\textwidth}
\begin{pseudocode}\\
{Anticoord}
Anticoord(i::ID,myb::BIT,otherb::BIT)\\
\\
\mbox{Randomly choose a bit }myb \in \{0,1\}\\
\mbox{communicate }myb\mbox{ to the }\\
\mbox{other player as its }otherb.\\
\mbox{\textbf{if} }[myb\oplus otherb \equiv i \mbox{ (mod 2)}] \\
\hspace{5mm}\mbox{ \textbf{then} play }C\\
\hspace{5mm}\mbox{ \textbf{else} play }S\\
\label{alg2.2}
\end{pseudocode}
\end{minipage} 

\vspace{2mm} 
\vspace{-2mm}


The intuitive conclusion of these two examples is simple: Definition~\ref{def:basic} is not sufficient. Some simple games may have coordinated \textbf{protocols} that could properly be called "Kantian". In this section \textbf{we give a somewhat more general  definition\footnote{other plausible alternatives are discussed (and ruled out) in Section~13 of the longer version \cite{istrate-kantian-tr}. Of special interest is the relation between Kantian equilibria and \emph{team-reasoning equilibria}.}  of Kantian equilibria, but not for general games, only for a class of "symmetric" games}. There are multiple definitions of game symmetry in the literature \cite{ham2013notions,tohme2019structural}; the most important one requires that for every player $i$, action profile $(x_1,x_2,\ldots, x_n)$ and permutation $\sigma \in S_{n}$, we have $u_{\sigma(i)}(x_1,x_2,\ldots, x_n)= u_{i}(x_{\sigma(1)},x_{\sigma(2)},\ldots, x_{\sigma(n)})$. We we use a slightly less demanding definition (we call our version \emph{Pareto symmetry}, see Definition~16 in the longer version~\cite{istrate-kantian-tr}) to capture some asymmetric games like Romer's version of BoS. This game is asymmetric, but in an inconsequential way: all the asymmetries concern dominated profiles. 
Our extension is inspired by the concept of  \textbf{program equilibria} \cite{howard1988cooperation,tennenholtz2004program}. To define these equilibria we associate to every finite normal-form game an extended game whose actions correspond to \emph{programs}. Agents' programs have access to own and others'  sources\footnote{this aspect was important for Nashian optimization. It will be less so for us, since deviations are not present in the definition of Kantian equilibria. On the other hand, a Kantian agent playing program $P$ can make sure it is not taken advantage upon by the other players, either alone,  as in \cite{tennenholtz2004program}, by reading other players' programs and only playing $P$ when all do, or with the help of a \emph{mediator}, which implements on behalf of all players the following protocol: if all agents follow $P$ then the mediator will simulate $P$ on behalf of the  agent; otherwise it will play in a Nashian way. },  and can act on them. A program equilibrium is a Nash equilibrium in the extended game. We extend this idea to Kantian equilibria: 
\begin{definition} 
Given a (Pareto symmetric) game $G$ with identical actions sets for all players, a \emph{Kantian program equilibrium in $G$} is a probability distribution $p$ on the action profiles of $G$ such that: (a). $p$ has its support on the set of action profiles that are \emph{Pareto optimal}, that is not Pareto dominated by any other action profile. (b). $p$ is implemented by \textbf{agents playing  a common program $P$} in the extended game. (c). there exists no probability distribution $q$ with properties a. and b. such that the vector of expected utilities $(E[u_{i}(q)])$ Pareto dominates the vector $(E[u_{i}(p)])$. 
\label{simpe}
\end{definition} 
The  assumption that \emph{players have the same action sets} is motivated by the "common program" requirement of Def.~\ref{simpe} (b).
Point (a). encodes a simple rationality condition. Points (b). and (c). embody a generalized version of the Kantian categorical imperative: (b). encodes the constraint of identical behavior, (c). encodes the fact that implementing $p$ is "a best action" for all players.

\textbf{We only need to formalize what we mean by \textit{program} in this definition.} The semantics is inspired by the one in \cite{tennenholtz2004program}, but the full formalization is somewhat subtle. We defer a full presentation to Section~15 of the longer version \cite{istrate-kantian-tr}. A couple of technical points are, however, worth stating:  
\begin{itemize}
\item[-] The semantics of programs in \cite{tennenholtz2004program} does not allow for any synchronization between different versions of the same program, other than testing whether they are syntactically equivalent. Since (as recognized in Theorem~\ref{miscoord}) we need to include correlated randomness, we need to extend the semantics of programs from \cite{tennenholtz2004program}, where this is not possible. There are many ways to do it, but one way is to allow \emph{correlated sampling from distributions}: all the programs get an identical sample from a given distribution. 
\item[-] \textbf{However, simply adding correlated sampling of \emph{action profiles} to the semantics of programs from \cite{tennenholtz2004program} leads (see Theorem~9 in the long version \cite{istrate-kantian-tr})  to paradoxical results:} every convex combination of Pareto optimal strategy profiles would be a Kantian program equilibrium. This is an issue: in Prisoners' Dilemma profiles $(C,D),(D,C)$ could perhaps be justified from a "team reasoning" perspective where one player "sacrifices itself" so that the other one walks free. To accept them as "Kantian equilibria" seems, however, problematic (see also the discussion in sections~13.3 and~14 of \cite{istrate-kantian-tr})
\item[-] If, on the other hand, \textbf{agent programs didn't communicate at all, used no private randomness, or used no specific ID/payoff information} then they would run identically for all agents, coordinating on the same  action (excluding, thus, scenarios like that of Example~\ref{ex}, that we want to model). 
\item[-] We will take a middle-ground approach, and assume that agents can use their ID and the information about the game payoffs in a very limited way, that makes the program act "identically with respect to a group of symmetries acting transitively on the set of agents". This requires us to restrict ourselves to the class of Pareto symmetric games of Definition~16 in the longer version \cite{istrate-kantian-tr}, whose set of Pareto dominant action profiles has such symmetries. The precise technical details are spelled out in Section 15 of the longer version \cite{istrate-kantian-tr}, where we prove (Theorem~10) a characterization of Kantian equilibria for Pareto symmetric games which also shows that Kantian program equilibria are a strict subset of the class of team-reasoning equilibria.

\end{itemize}

Algorithms 4.1 and 4.2 lend some credibility to the intuition that Kantian equilibria are somehow related to some "symmetric" notion of correlated equilibria. This intuition is correct: in Definition~19 in  the longer version \cite{istrate-kantian-tr} we define a notion of "correlated symmetric equilibrium". We then prove:  

\begin{theorem} 
Correlated symmetric equilibria of symmetric games are Kantian program equilibria. 
\label{kant-corr}
\end{theorem} 

Kantian program equilibria allow players to obtain a better expected payoff in Platonia Dilemma: 
\begin{theorem} 
Algorithm~\ref{alg3} implements a Kantian program equilibrium for Platonia Dilemma.
\label{kantian-pd}
\end{theorem} 
\begin{proof} 
Points (a). and (b). from the definition of Kantian program equilibria are clear, the only one that merits a discussion is point (c). 

The expected utility of each player under Algorithm~\ref{alg3} is equal to $1/n$. Since the sum of utilities of all players under a particular set of random choices is equal to 1, no vector of expected utilities can strictly dominate the vector $(1/n,1/n, \ldots, 1/n)$ of expected utilities for the Algorithm. 
\end{proof} 


\begin{pseudocode}{Choose-Winner}{i,b_1,b_2,\ldots, b_{n}}
\mbox{Randomly choose an integer }b_{i}\in \mathbb{Z}_{n}\\ 
\mbox{\textbf{if} }[\sum_{j=1}^{n} b_{j}\equiv i\mbox{ (mod n)}]\\
\hspace{5mm}  \mbox{ \textbf{then} S(UBMIT)} \\
\hspace{5mm}\mbox{ \textbf{else} D(ON'T)}\\
\label{alg3}
\end{pseudocode}


\section{Some computationally efficient other-regarding equilibria}
\label{sec:efficient} 

As defined in the previous section, Kantian program equilibria for games with identical action sets inherit some of the definitional problems of "ordinary" program equilibria. Among them: 
\begin{itemize}
\item[-] \emph{fragility}: (Kantian) program equilibria are sensitive (see e.g. \cite{oesterheld2019robust}) to the precise specification of programs: do we insist that all agent programs are syntactically identical, or just "do the same thing"? See \cite{lavictoire2014program,van2013program} for some attempted solutions for program equilibria that could be adapted to our setting.   
\item[-] \emph{lack of generality}: Definition~\ref{simpe} it is only applicable to (some of the) games with identical action sets. To further generalize it to all finite normal-form games one would need to specify what it means for two agents to "take the same course of action" in settings with differing action sets. 

\item[-] \emph{lack of predictive power}: There may be multiple (even infinitely many) Kantian program equilibria. 

\end{itemize}

Given these objections, and with constraints (I)-(IV) in mind, we propose in the sequel a substantially more modest approach: Rather than seeking a general definition of Kantian equilibria we propose instead  \emph{several} other-regarding equilibria. \textbf{They all correspond intuitively to real-life situations, are tractable, can be justified by team reasoning and are related, for symmetric coordination games, to Kantian equilibria.}  One was independently suggested in \cite{kordonis2016model}, the other ones are first introduced here: 

\begin{definition} 
A \emph{Rawlsian equilibrium} is a probability distribution over Pareto optimal profiles maximizing the 
\emph{egalitarian social welfare} (the expected utility of the worst-off player) and is strictly dominated by no other profile with this property.  Such equilibria implement the idea of justice as fairness \cite{rawls2001justice}.
\end{definition} 
\begin{example} 
We modify the BoS example as in Fig.~\ref{multiple} (c): perhaps 1 is a classical music lover, that gets a higher utility than the other player by going, together with its partner, to any of the two concerts. Then (S,S) is the (unique) Rawlsian equilibrium. Choosing such an equilibrium is an example of altruistic behavior from player 1, since it maximizes the payoff of its non-music-lover partner.  
\label{bos2} 
\end{example} 
\begin{definition} 
A \emph{Bentham-Hars\'anyi equilibrium} is a probability  distribution on Pareto optimal profiles maximizing the sum of expected payoffs.  See \cite{harsanyi1955cardinal} for a philosophical motivation. 
A \emph{best-off equilibrium} is a prob. distrib. on Pareto optimal profiles maximizing the largest expected payoff, and strictly dominated by no profile with this property. E.g., in Exp.~\ref{bos2} (B,B) is the unique Bentham-Hars\'anyi/best-off equilibrium. 
\end{definition}

Although a best-off equilibrium may not seem "fair", there exist real-life "team reasoning" situations that elicit behavior suggestive of such an equilibrium: one such example is, for instance, scenarios where members of a team "sacrifice" for one of their members (e.g. parents for a child).  

The equilibrium notions we introduced so far implicitly assumed that player utility is given by material payoffs. Sometimes the frustration a player feels is derived by counterfactually comparing its realized payoff with all possible ones. There are many implementations of this idea. The following notion quantifies the extent to which a given profile is worse for the given player than a random profile. 

\begin{definition} 
The \emph{percentile index of profile $a$ for player $i$} is the percentage of Pareto optimal profiles that would get  $i$ a strictly better payoff than $a$. 
A \emph{Rawlsian percentile equilibrium} is a profile minimizing the largest expected percentile index of all players, and strictly dominated by no profile with this property. 
\end{definition} 

\begin{example} Consider the game shown in Figure~\ref{multiple} (b). 
Then percentile indices of Pareto optimal profiles are $(0,100)$ for $(C,C)$, and $(100,0)$ for $(D,D)$, respectively. Profile $\frac{1}{2}(C,C)+\frac{1}{2}(D,D)$ is a Rawlsian percentile  equilibrium. Player 1 gets average utility 7 while player 2 gets average utility $\frac{3}{2}$. 
\label{ex}
\end{example} 

\noindent An even less cognitively sophisticated model of agent frustration relies on classifying outcomes as "happy/not happy". The following is a simple example of such a notion: 

\begin{definition} The \emph{natural expectation point of player $i$} is the median (over all undominated pure strategy profiles) payoff. If there are two medians then the average value is taken. 
A player is \emph{happy in a pure strategy profile $a$} iff its payoff is larger or equal than its natural expectation point and  \emph{unhappy} otherwise. 

An \emph{aspiration equilibrium} is a mixed strategy profile that minimizes the largest probability of unhapiness among all players and is strictly dominated by no other profile with this property. 
\end{definition} 

\begin{example} 
Take a coordination game with payoffs $(C,C)\rightarrow (10,1)$, $(D,D)\rightarrow (9,2)$, $(E,E)\rightarrow (8,3)$, $(F,F)\rightarrow (4,7)$.  
The natural expectation points of players are $8.5$ and $2.5$, respectively. The first player is happy in $(C,C)$ and $(D,D)$, the second in $(E,E)$, $(F,F)$. Hence in $\frac{1}{4}(C,C)+\frac{1}{4}(D,D)+\frac{1}{4}(E,E)+\frac{1}{4}(F,F)$ the players are happy $50\%$ of the time and no mixed action profile can do any better. 
\label{bos4}
\end{example} 
Unlike general Kantian program equilibria, the equilibria we defined are computationally tractable: 

\begin{theorem} 
Rawlsian, Rawlsian percentile, Bentham-Hars\'anyi, best-off, aspiration equilibria exist
and can be found by solving a sequence of linear programs (hence in polynomial time). 
\label{thm:efficient} 
\end{theorem}

We now connect our other regarding equilibria to Kantian equilibria in symmetric coordination games. 
We call an equilibrium point
\emph{extremal} if it cannot be written as a nontrivial convex combination of other (similar) equilibria.  
We show that extremal self-regarding equilibria generalize Kantian pure equilibria. Extremality is needed, since our equilibria are closed under convex combinations (such combinations are justifiable from a magical thinking perspective, see footnote 7), while pure Kantian equilibria are not. Because of Thm.~\ref{nphard} no similar connection is likely for mixed Kantian equilibria: 

\begin{theorem} 
In symmetric diagonal games Rawlsian, Bentham-Hars\'anyi, best-off, Rawlsian percentile, aspiration equilibria coincide with convex combinations of Kantian pure equilibria. 
\label{symmetric} 
\end{theorem} 
\vspace{-5mm}
\section{Agents with bounded greed} 
\label{sec:bounded} 
So far we have assumed that people are other-regarding. In reality people are not unrestricted optimizers, nor are they perfect Kantian moralists.  Alger and Weibull \cite{alger2013homo} attempted to interpolate between utilitarian agents and Kantian ones, by defining \emph{homo moralis} to be an agent whose utility has the form $u_{i}(x,y)=(1-k)\pi(x,y)+k\pi(x,x)$, where $k\in [0,1]$ is the so-called \textit{degree of morality} of the agent. They showed that evolutionary models with assortative mixing and incomplete information favor a particular kind of homo moralis, those whose degree of morality coincides with the degree of assortativity of the matching process. Interesting as this result is, it has some weaknesses. For instance \cite{alger2013homo}, homo moralis behaves like homo economicus in Prisoners' Dilemma and all constant-sum games when $k\neq 1$. In other words, agent behavior is not sensitive to the degree of morality, as long as the agent is not Kantian.

We give (for symmetric games, but the idea can be extended to general ones, via Kantian program equilibria) a definition with the same overall intention, but capturing a slightly different agent behavior: 
\begin{definition} 
Let $\lambda \in [1,\infty]$. Agent $i$ is called \emph{$\lambda$-utilitarian} if, for every action profile $(a_i,b)$, its utility $u_{i}(a_{i},b)$ is  (a). $\pi_{i}(a_{i},(\overline{a_{i}})_{-i})$  if  $a_i$ is a Kantian action. (b). 0 if $a_{i}$ is not Kantian and $\pi_{i}(a_{i},b)\leq \lambda\cdot  \pi_{i}(X^{OPT})$; (c). $\pi_{i}(a,b)$ if $a_{i}$ is not Kantian and $\pi_{i}(a_{i},b)\leq \lambda\cdot  \pi_{i}(X^{OPT})$. I.e., a $\lambda$-utilitarian agent deviates from its Kantian action $X^{OPT}$ \textbf{only} if the utility it obtains  is more than $\lambda$ times larger. 
\label{part-kant} 
\end{definition} 
We call the number $\frac{1}{\lambda-1}$ \emph{the greed index} of $i$. It varies between 0 (Kantian agents) and $\infty$ (purely utilitarian ones).  The natural equilibrium concept for such agents is no longer Kantian, but Nash equilibrium. 
Definition~\ref{part-kant} allows giving an empirically plausible justification of all possible outcomes in PD: 
\begin{theorem} 
All pure action profiles in PD are Nash equilibria of agents with varying degrees of greed. 
\label{pdgreed} 
\end{theorem} 
\begin{proof} 
Bounded-greed agents still coordinate on the Kantian equilibrium $(C,C)$ as long as both their greed indices are $<2$ (i.e. they would need at least a twofold increase in payoff to deviate). If one of them has greed index $<2$ and the other one has greed index $\geq 2$, then the latter one will defect. If both agents have greed indices $\geq 2$, then they will coordinate, just as if utilitarian agents would do, on the Nash equilibrium $(D,D)$. 
\end{proof}


\section{Conclusions} 

Our main contribution is bringing Kantian equilibria (and related concepts) to the attention of TARK community, showing that this notion is theoretically interesting, but that the road to implementable behaviors probably goes through less general equilibrium concepts. 
Many of the notions we introduced, on the other hand, 
including Kantian program equilibria and bounded greed agents, deserve further investigation. For instance a justification like that of Theorem~\ref{pdgreed} could be used as a rationality criterion. One could look for evolutionary justifications of bounded greed agents along the lines of \cite{alger2013homo}. One could use such agents in relation to work on the concept of \emph{price of anarchy} \cite{selfish-routing}. On a more conceptual level, the use of \emph{frames} in game theory \cite{bacharach2006beyond,bermudez2021frame} and how this interacts with equilibrium notions deserves further study.  Finally, several open problems remain: Can we find algorithms for our equilibria that bypass the need for solving multiple LP's? Is the problem from Theorem~\ref{nphard} NP-complete (i.e. in NP)?

\bibliographystyle{eptcs}
\bibliography{generic}

\end{document}